\documentclass[a4paper,10pt]{article}

\usepackage{babel}
\usepackage[utf8]{inputenc}
\usepackage[T1]{fontenc}
\usepackage{amsmath, amsfonts, amssymb, amsthm}
\usepackage{mathtools} 
\usepackage{graphicx,subcaption}
\usepackage{placeins} 
\usepackage{hyperref}
\usepackage{lmodern}
\usepackage{tikz,calc}
\usepackage{cite}
\usetikzlibrary{calc}

\newtheorem{thm}{Theorem}
\newtheorem{definition}{Definition}

\newtheorem{lem}{Lemma}

\newtheorem{rem}{Remark}
\newtheorem{exa}{Example}

\newcommand{\Z}{\mathbb{Z}}

\graphicspath{{./figures/}}

\newcommand{\overbow}[1]{
   \tikz [baseline = (N.base), every node/.style={}] {
      \node [inner sep = 0pt] (N) {$#1$};
      \draw [line width = 0.4pt] plot [smooth, tension=1.3] coordinates {
         ($(N.north west) + (0.1ex,0)$)
         ($(N.north)      + (0,0.5ex)$)
         ($(N.north east) + (0,0)$)
      };
   }
} 

\newcommand{\cycle}[1]{\overbow{#1}} 

\newcommand{\cplus}[2]{#1\hspace*{-0.04cm}{+}\hspace*{-0.04cm}#2}
\newcommand{\cmin}[2]{#1\hspace*{-0.04cm}{-}\hspace*{-0.04cm}#2}

\newcommand{\support}[1]{\mathrm{supp}(#1)}
\newcommand{\comment}[1]{{\footnotesize\color{gray}{:#1}}\normalsize}

\title{Ants on the highway}
\author{Anahí Gajardo, Victor H. Lutfalla, Michaël Rao}
\date{Departamento de Ingenier\'ia Matem\'atica \& Centro de Investigación en Ingeniería Matemática, Universidad de Concepci\'on,
Casilla 160-C, Concepci\'on, Chile\\
\texttt{anahi@ing-mat.udec.cl}\\
Aix-Marseille Univ, CNRS, I2M, Marseille, France\\
\texttt{victor.lutfalla@math.cnrs.fr}\\
CNRS\& LIP, Ecole Normale Supérieure de Lyon, France\\
\texttt{michael.rao@ens-lyon.fr}\\
}
\begin{document}
\maketitle

\begin{abstract}
We perform intensive computations of Generalised Langton's Ants, discovering rules with a big number of highways.
We depict the structure of some of them, formally proving that the number of highways which are possible for a given rule does not need to be bounded, moreover it can be infinite.
The frequency of appearing of these highways is very unequal within a given generalised ant rule, in some cases these frequencies where found in a ratio of 1/10$^7$ in simulations, suggesting that those highways that appears as the only possible asymptotic behaviour of some rules, might be accompanied by a big family of very infrequent ones.
\end{abstract}


\section{Introduction}

Introduced in the 80's independently by different researchers~\cite{Lang86, KongCohe91JSP, DewdTurk}, the automaton mostly known as \emph{Langton's ant} remains intriguing.
It is cited in vulgarisation media as a paradigm of `emergent phenomena', `unpredictability', `abnormal diffusion' and as evidence of the impossibility of a `theory of everything'.

Its \emph{emerging phenomenon} consists in a pattern that makes the ant to move in such a way that after 104 steps, it is shifted by (-2,2) and the same pattern appears again in the configuration, also shifted.
Given this, the ant movements will repeat again and again, producing a periodic behaviour known as the \emph{highway} (see figure~\ref{fig:exs}(a)).
In the terminology of~\cite{Kurk}, this behaviour is due to the presence of a periodic point in the \emph{Moving Tape Model}, with the particularity that in the propagating direction, the background is homogeneous.

The intriguing fact is that this pattern spontaneously seems to appear on every simulation over initially finite configurations, the transient phase vary mostly from 1000 to 100,000 iterations for initial configurations with only one black cell in a radius of 3 cells around the ant position; and until now, no simulation suggest the conjecture might be false.

Several generalisations have been proposed: with more movements~\cite{BuniTrou92}, with more states~\cite{GaleProp94}, on other grids~\cite{GGM02}, and on other dimensions~\cite{DorbGaja}.
The generalization
that seems to better preserve the Langton's ant properties is the one we call \emph{generalised ants}~\cite{GaleProp94}.
In this class, emergent behaviours are also observed, with different periods, or different shapes.
In addition, there are rules where no highway or any similar asymptotic behaviour arrives for some initial finite configurations.

Motivated by the fact that computational simulations can show only a very restricted part of the whole (and infinite) space of possibilities, and that on such complex systems strange phenomena may appear only on extremely big initial configurations, we started a series of high performance simulations, looking for new behaviours that could put in evidence this limitation of the computational tools.

Section~2 introduces the model and makes a journey over the main theoretical founding about generalised ants since their discovery; also it formally defines the concept of ``\emph{highway}''.
Section~3 resumes what simulations gave us.
In Section~4, we establish our first theoretical result: {\em there is a family with a growing set of highways}; while our second result is demonstrated in Section~5: {\em there is a generalised ant with an infinite set of highways.}
We leave some reflections and overtures in Section~6.

\section{Generalised ants}\label{sec:defs}

Langton's ant is an agent, provided with an orientation in the plane $\{\rightarrow,\uparrow,\leftarrow,\downarrow\}$ that moves and turns either to the right or to the left depending on the colour of its underlying cell, which it flips after leaving the cell.
Cells can have only one of two colours, \emph{white} or \emph{black} (0 or 1). Over a white cell, the ant turns to the left and over a black cell the ant turns to the right.
Generalised ants have more colours, its movement rule is given by the meaning of each colour, and this one changes following a simple increasing cycle.

\begin{definition}{\cite{GaleProp94}}
A \emph{generalised ant} with rule word $w\in\{L,R\}^+$ is an automaton with state set $Q=\{\rightarrow,\uparrow,\leftarrow,\downarrow\}$, representing the four unitary vectors of $\Z^2$, for example $\leftarrow\ = (-1,0)$,  and alphabet $A = \{0,1,...,|w|-1\}=\Z/\Z_n $, that moves over the grid $\Z^2$.
A configuration is an element from $A^{\Z^2}\times \Z^2\times Q$, that represents:

\begin{itemize}
\item an assignment of symbols from $A$ to each cell in $\Z^2$ called \emph{picture},
\item a marked cell representing the \emph{position} of the ant, and
\item the current \emph{state} of the ant.
\end{itemize}

For a given configuration $(C,(i,j),d)\in A^{\Z^2}\times \Z^2\times Q$, the global transition function $T_w$ is defined by $T_w(C,(i,j),d)=(C',(i',j'),d')$, where:
\begin{itemize}
\item $C'(k,l)=\left\{\begin{array}{ll}
  C(k,l)+1\mod |w| & \text{if } (k,l)=(i,j)\\
  C(k,l)   & \text{otherwise} 
\end{array}\right.$,
\item $d'$ is a rotation of $d$ by $90^\circ$ clockwise if $w_{C(i,j)}=R$, or counterclockwise if $w_{C(i,j)}=L$,
\item $(i',j')=(i,j)+d'$.
\end{itemize}
\end{definition}

We call \emph{trace} of ant $w$ from configuration $(C,(i,j),d)$ the sequence of symbols that the ant encounters when evolving from configuration $(C,(i,j),d)$.
A trace contains all the necessary information to recover the initial picture, up to rotations, at least on the set of visited cells~\cite{GajaJAC}.

We assume that the rule word has at least one $L$ and one $R$, otherwise we call it \emph{trivial}, because in that case its behaviour will simply consist in constantly turning in a fixed sense.

An ant simulator is provided at \href{https://lutfalla.fr/ant/}{\texttt{lutfalla.fr/ant/}}, most of the results we present in this paper are illustrated there.

Very few theoretical results exists about this system.
The first remark that should be done is that the dynamics is reversible, that is, the transition function $T_w$ is bijective, thus no information is lost through the time.
After this, the most important result about ants establishes that: \emph{independently from the initial configuration, the ant orbit will always be unbounded.}
This result was formally proved by Bunimovitch and Troubetzkoy in~\cite{BuniTrou92} for the rule \emph{LR}, but it is easy to see that it is also true for the whole family of non-trivial ants.
The proof is very simple but clever, and it is based on a very strong property of this system: if the ant starts with an horizontal direction over a cell $(i,j)$, then it will always be with horizontal direction over every cell $(i',j')$ such that $i+j=i'+j'\mod 2$, and with vertical direction otherwise (Figure~\ref{fig:chess}).
This is because generalised ants always alternate between vertical and horizontal directions, and also the parity of the cell where they lands is always alternating between odd and even, which couples these two situations.

The most important consequence of this fact is that each cell has only two fixed entering sides: the ant can only enter a horizontal cell from east or west, and it can only enter a vertical cell from north or south. 
Then if an ant trajectory was finite, the set of visited cells should contains corners, that is cells with only two adjacent cells inside the set, and these forming an angle of $\pi/2$.
Independently on the orientation of this cell, horizontal or vertical, it has only one entering cell inside the set and also only one output, which is either to the left or to the right of the first, this imply that the cell can be only visited finitely many times, 
since the rule word contains at least one $L$ and one $R$.
But this is impossible, because the system is reversible, thus any orbit on a finite set is periodic,  
and the corner must be visited infinitely many times.

As a corollary, if in a given trajectory we consider the set of cells which are visited infinitely many times, we know that this set (if not empty) cannot have any corner.

  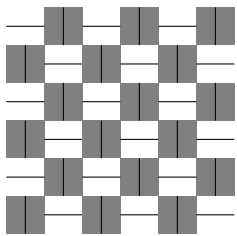
\begin{figure}[htb]
    \center
    \begin{tikzpicture}[scale=.5]
    \foreach \x in {0,2,...,4}{
        \foreach \y in {0,2,...,4}{
            \fill[gray] (\x,\y) rectangle ++(1,1);
            \fill[gray] (\x+1,\y+1) rectangle ++(1,1);
            \draw (\x+.5,\y) -- (\x+.5,\y+1);
            \draw (\x+1.5,\y+1) -- (\x+1.5,\y+2);
            \draw (\x+1,\y+.5) -- (\x+2,\y+.5);
            \draw (\x,\y+1.5) -- (\x+1,\y+1.5);
            }
        }
    \end{tikzpicture}
    \caption{Each cell has two fixed entering sides: either vertical or horizontal.}
    \label{fig:chess}
  \end{figure}

Up to our knowledge, the class of generalised ants was introduced by the first time in~\cite{GaleProp94}, where different behaviours are presented. 
Since then, other researchers started making simulations, finding quite surprising behaviours~\cite{BeurToma97}.
Figure~\ref{fig:exs} shows nice rules taken from
\href{https://mathtician.weebly.com/langtons-ant.html}{\texttt{\small mathtician.weebly.com/langtons-ant.html}}.

  \begin{figure}[htb]
    \center
    \begin{tabular}{cccc}
    \includegraphics[width=0.2\textwidth]{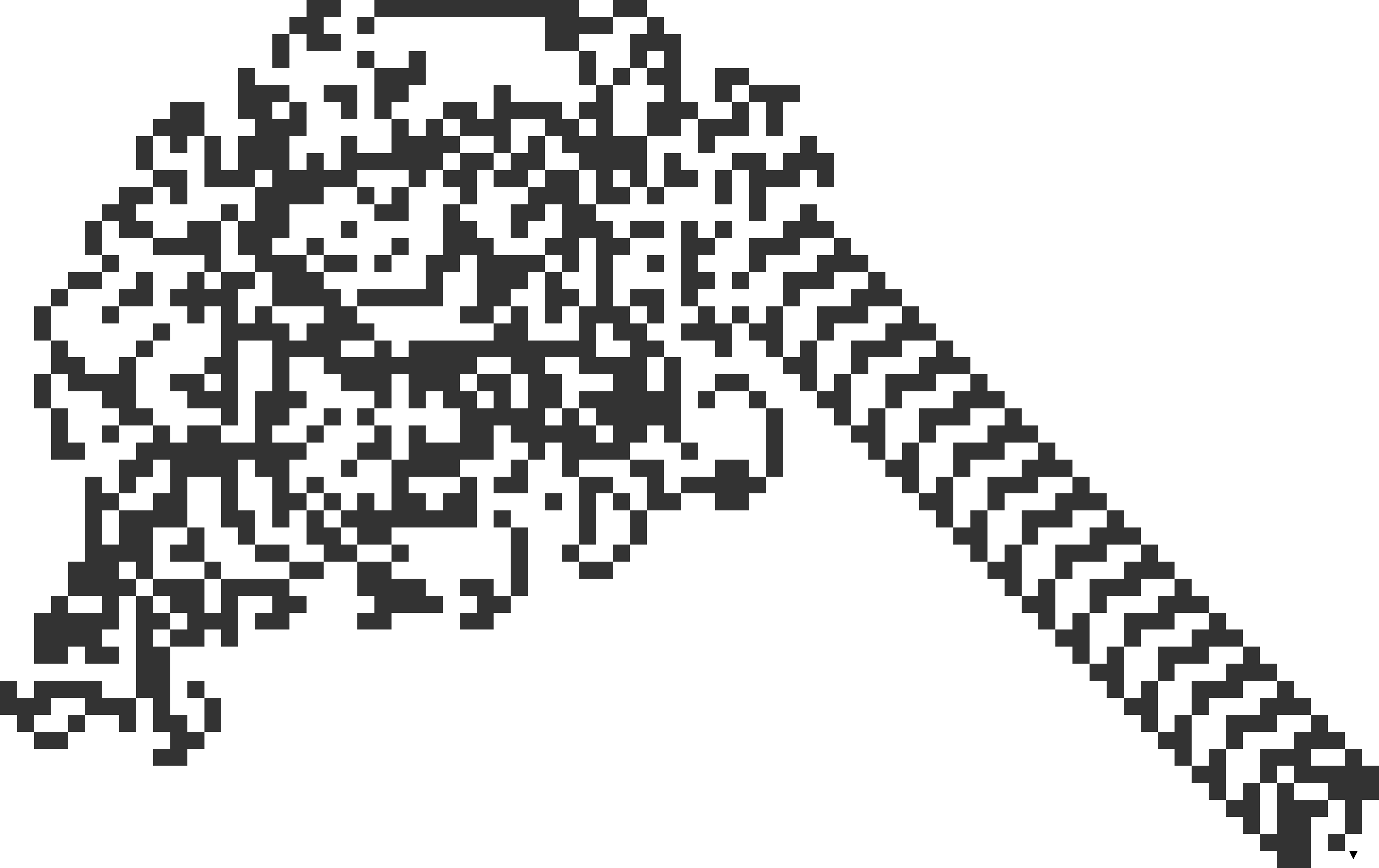} &
    \includegraphics[width=0.2\textwidth]{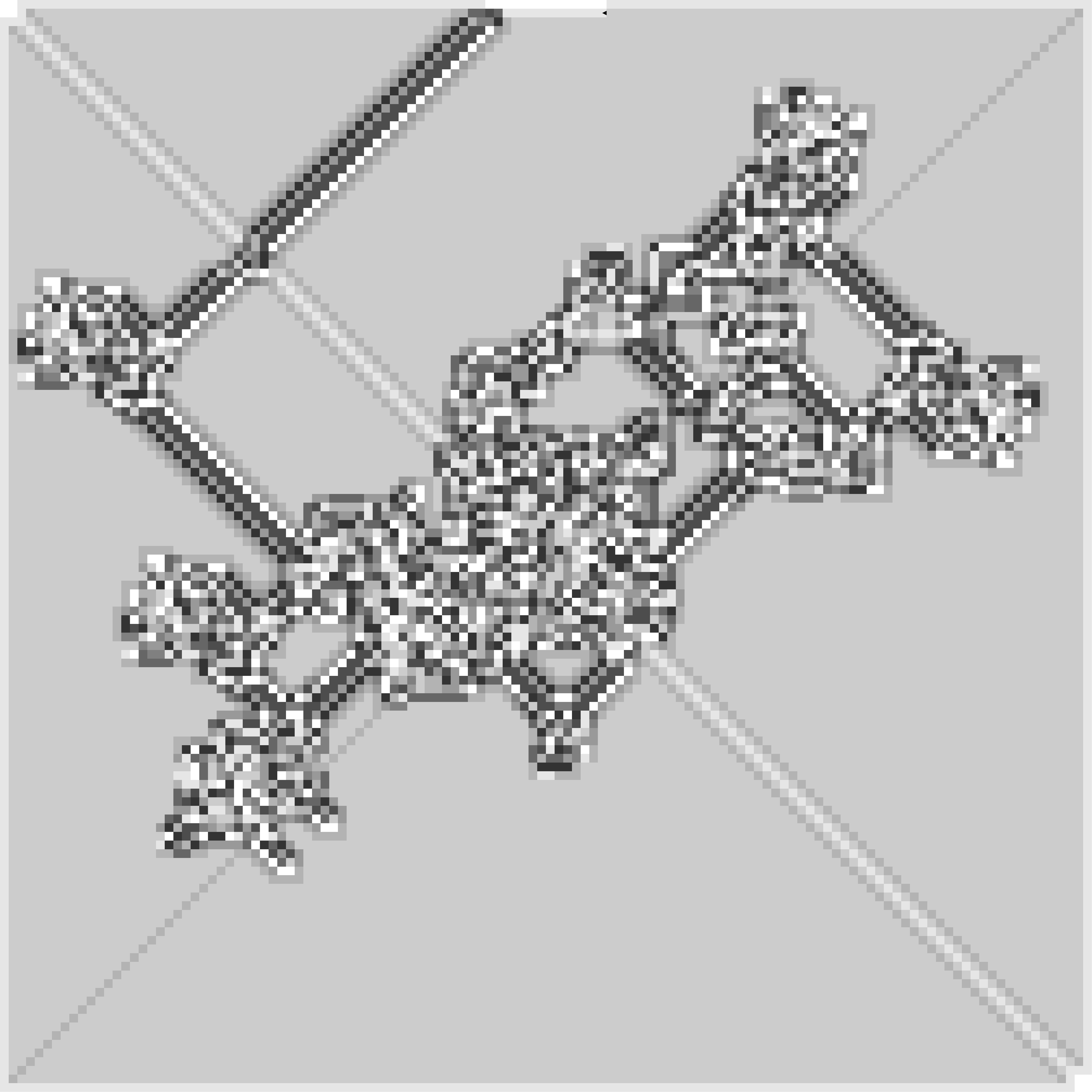} &   
    \includegraphics[width=0.2\textwidth]{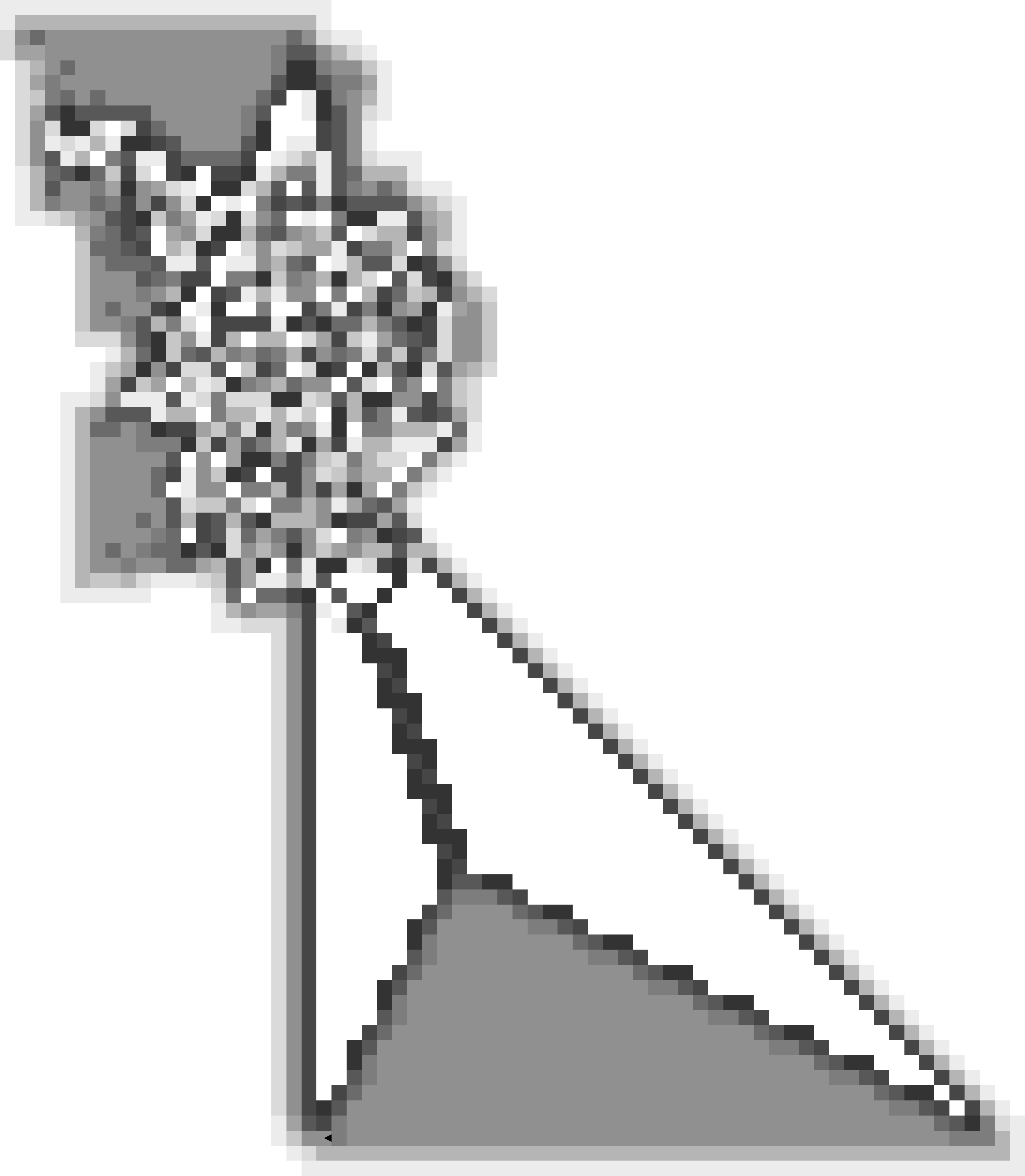} &  
    \includegraphics[width=0.2\textwidth]{zoology_ant_LRRRRLLLRRRL.pdf}    \\
    (a)&(b)&(c)&(d)
    \end{tabular}
    \caption{Zoology of emergent behaviours of generalised ants: (a) rule \emph{LR}, the original rule; (b) rule \emph{LRRRRRLLR}  builds a textured square over which highways repetitively born and crash; (c) rule \emph{LLRRRLRRRLLL} builds a growing triangle that travels on the plane; (d) rule \emph{LRRRRLLLRRRL} builds a square with a logarithmic spiral inside.}
    \label{fig:exs} 
  \end{figure}
  
Another beautiful result is proved in~\cite{GPST95}: \emph{All non-trivial rules whose code word $w$ belongs to $\{LL,RR\}^+$ will produce an evolution composed by an infinite sequence of closed trajectories that always comes back to the initial cell if the initial configuration satisfy the appropriate conditions}~\cite{GPST95}.
This result has two consequences: on the initially white configuration the ant evolution produces patterns with bilateral symmetry, as butterflies (figure~\ref{fig:LLRR}); and also, over this set of initial configurations, no \emph{highway} will ever appear.

  \begin{figure}[htb]
    \center
    \includegraphics[width=0.4\textwidth,angle=180]{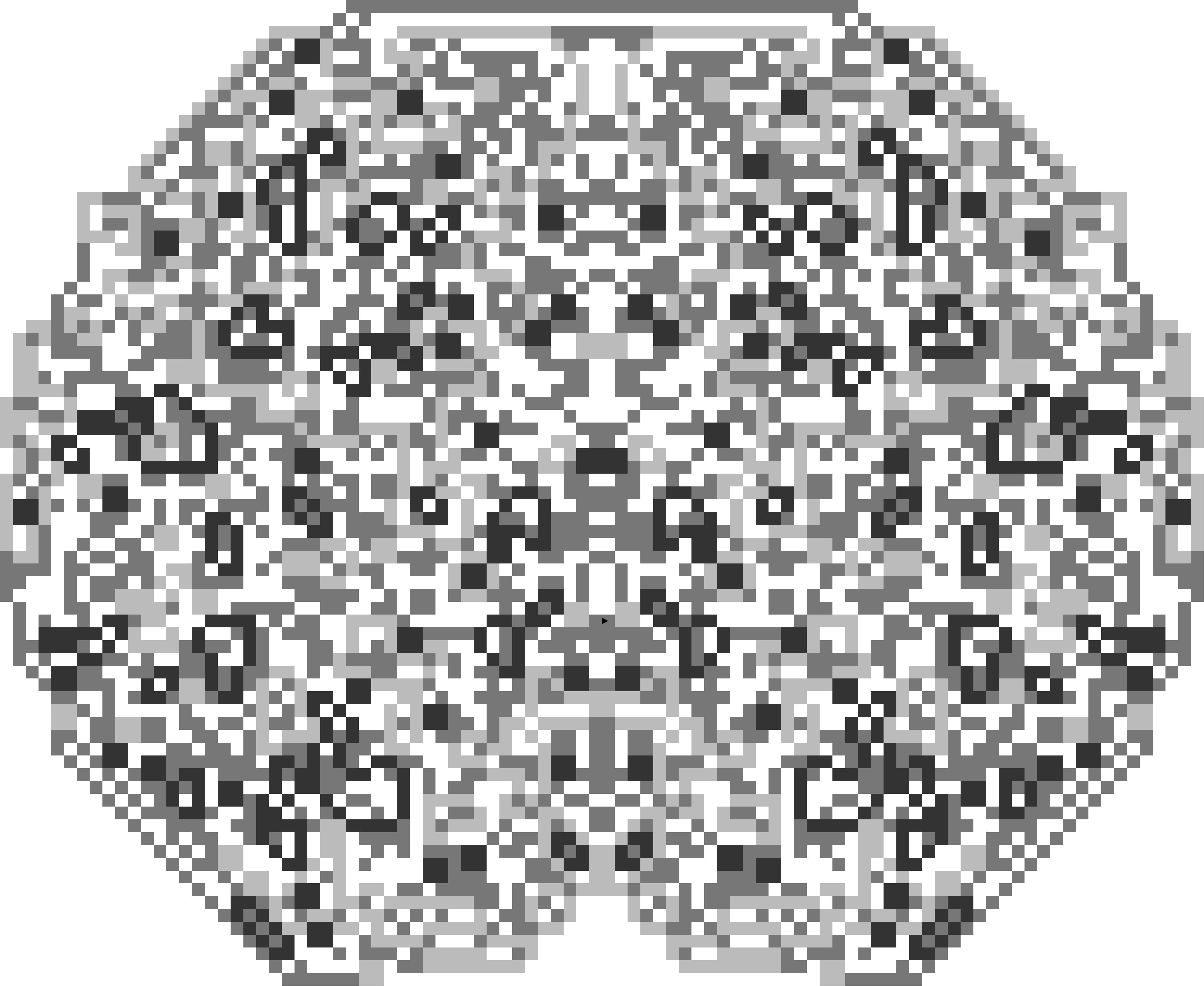}
    \caption{Configuration at iteration 817888 of rule LLRR starting with the initial configuration where all the cells are white (0).}
    \label{fig:LLRR}
  \end{figure}
  
The last known result about this class shows that predicting the trajectory of a non-trivial generalised ant is P-complete and some predictions are undecidable. More precisely, any non-trivial ant rule can compute logical circuits if the appropriate initial configuration is given. Moreover, there exists a periodical initial configuration such that, the dynamics of the ant can simulate a universal Turing machine over any finite input, if the appropriate finite perturbing pattern is applied~\cite{MGHM}.
This result somehow explains why the ant analysis is so hard.

In this paper we are interested in highways, that is, emergent patterns that makes the ant to follow a periodical movement that will propagate over an homogeneous background.
In order to formalise these phenomena, we need to introduce some notations and concepts.

We say that a configuration $(C,(i,j),d)\in A^{\Z^2}\times \Z^2\times Q$ is \emph{finite} if the set of cells with a symbol different than 0 is finite: $|\{(a,b)\ :\ C(a,b)\neq 0\}|<\infty$, in this case we also say that the picture $C$ is finite.
The picture full of 0s is called \emph{white}.

A \emph{pattern} is an assignment $P$ of symbols to a finite set of cells $S$, that is, a function $P:S\rightarrow A$ with $S\subseteq \Z^2$. The set $S$ is the \emph{support of $P$}, and it is denoted by $\support{P}$.
A restriction of a picture $C$ to a finite set $S$ is \emph{the pattern of $C$ on $S$}.

The transition function $T_w$ of a generalised ant can be extended to patterns if also a position and direction inside the pattern support is given, that is: $T_w$ may be applied to a tuple $(P,(i,j),d)$ if and only if $(i,j)\in \support{P}$.

Finally, patterns, pictures and configurations can be shifted by a vector $(a,b)\in\Z^2$.
The shift function $\sigma^{(a,b)}$ is defined on pictures, on configurations, on sets, and patterns by: 
\begin{itemize}
\item $\sigma^{(a,b)}(C)_{(i,j)}=C_{(i-a,j-b)}$, 
\item $\sigma^{(a,b)}(C,(i,j),d)=(\sigma^{(a,b)}(C),(i+a,i+j),d)$,
\item $\sigma^{(a,b)}(S)=S+\{(a,b)\}$,
\item $\sigma^{(a,b)}(P): \sigma^{(a,b)}(S)\rightarrow A$, $\sigma^{(a,b)}(P)_{(i,j)}=P_{(i-a,j-b)}$.
\end{itemize}

\begin{definition}
A \emph{highway} of period $N$ and drift $(a,b)$ of an ant rule $w$ is represented by a pattern $P$,  an ant position $(i,j)$ inside its support and an ant state $d$ such that the transition function $T_w$ can be applied $N$ times over it, after which, the new pattern $T^N_w(P)$ has the ant at position $(i,j)+(a,b)$ in state $d$ and the pattern $P$ appears inside $T^N_w(P)$ but shifted by $(a,b)$, more precisely:
\[ \forall (x,y)\in\support{P},\ P(x,y)=\left\{\begin{array}{ll} T^N_w(P)_{(x+a,y+b)}&\textrm{if }(x+a,y+b)\in\support{P}\\
0&\textrm{otherwise} \end{array}\right.\]
\end{definition}
This definition ensures that if the pattern appears in a configuration where all the cells in direction $(a,b)$ from $\support{P}$ are in state 0, then the same pattern will appear every $N$ steps and both the ant trace and ant movements will be $N$ periodic.
Behind, the cells in $\support{P}\setminus\support{\sigma^{(a,b)}(P)}$ will not be visited any more; repeated, they will configure the \emph{prints} leaved by the ant on the highway.

Of course, any of the patterns appearing during the evolution of a highway can be used to represent the same highway; we can understand a highway as a global phenomenon, or as a cyclic sequence of patterns.
Also, any rotation of a highway is the same highway, but with a rotated drift.

A highway can be characterized in terms of the trace of the ant, 
indeed if the trace of an ant that started on a finite configuration is eventually periodic, it is clear that the ant has entered a highway, because every periodic trace comes from a semiinfinite periodic pattern~\cite{GajaJAC}, and the only semiinfinite periodic pattern contained in an finite configuration is the white one. 
We use this characterisation to automatically detect highways in computer simulations.

\begin{exa}
The ant LR has a highway of period 104 and drift $(-2,2)$.
Figure~\ref{fig:h104} shows one of the patterns that represent it.
In gray are marked the non null cells, this smaller pattern is the one that we find shifted after 104 steps.
In light gray are drawn the cells with symbol 0.
If the background in direction $(-2,2)$ is full of 0s, the ant finds the whole pattern after 104 iterations, shifted by $(-2,2)$ and the history will repeat.
\begin{figure}[htb]
    \center
    \includegraphics[width=0.6\textwidth]{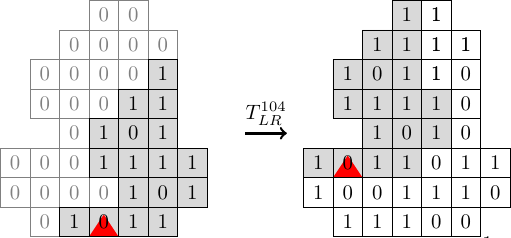}
    \caption{After 104 iterations the ant has moved by the vector $(-2,2)$, and the configuration which produced that movement is also shifted by $(-2,2)$.
    If the cells in direction $(-2,2)$ from this configuration have all the symbol 0, then this movement will repeat itself forever.}
    \label{fig:h104}
  \end{figure}
\end{exa}

\section{Simulations and classification}

Our goal is to tackle the highway conjecture.
In order to better understand the classical Langton's ant, we decided to study the highways and emerging behaviour of generalised ants.
The first step in this study was a large number of simulations of generalised ants over random perturbations of the $0$-uniform configuration.

For each computation, an initial pattern $P_0$ of size $11\times 11$ is taken uniformly at random, we then compute the trace $t$ of the ant for $10^5$ steps from the picture $C_0$ which is $0$-uniform with the pattern $P_0$ around the origin (and the ant at the centre of $P_0$).
The trace $t$ is then analysed to determine if a highway was reached.
For each ant, we ran the simulation at least $10^8$ times.
Table~\ref{tab:stats} shows the statistics of these simulations for some generalised ants which have remarkable properties.

For these simulation we are only interested in the length or period of the highways and not in the exact configuration, pattern or trajectory of the highway.
This means that multiple highways with different trace, drift or trajectory are counted as the same if they have the same period.

\begin{table}[htb!]
  \begin{tabular}{|l | r | p{.62\textwidth} |}
    \hline
    \textbf{Ant word} & \textbf{Highways} & \textbf{Period and frequency of the highways}\\
    \hline
    $LR$ & 100\% & 104\comment{100\%}\\
    $LLR$ & 100\% & 18\comment{100\%}\\
    $LLLR$ & 100\%& 52\comment{99.88\%},\ 156\comment{0.12\%}\\
    $LLLLR$ & 100\% & 34\comment{30.1\%},\ 68\comment{69.9\%}\\
    $L^5R$ & 100\% & 84\comment{100\%}\\
    $L^6R$ & 100\% & 50\comment{18.63\%},\ 100\comment{81.37\%}\\
    $L^7R$ & 100\% & 116\comment{100\%}\\
    $L^8R$ & 100\% & 66\comment{13.4\%},\ 132\comment{86.6\%}\\
    $L^9R$ & 100\% & 148\comment{100\%}\\
    $L^{10}R$ & 100\% & 82\comment{10.45\%},\ 164\comment{89.55\%}\\
    $L^{11}R$ & 100\% & 180\comment{100\%}\\
    $L^{12}R$ & 100\% & 98\comment{8.58\%},\ 196\comment{91.42\%}\\
    & & \\
    \hline
    $LLRL$ & 22.95\% & 308\comment{$3\cdot 10^{-4}\%$},\ 380\comment{$4\cdot 10^{-7}\%$},
                       384\comment{99.9997\%},\\
    & & 388\comment{$4\cdot 10^{-6}\%$}, 928\comment{$2\cdot 10^{-7}\%$}\\
    \hline
    $LLRLRL$ & 1.52\% & 220\comment{29.8\%},
                        244\comment{10.6\%},
                        268\comment{5.3\%},
                        284\comment{0.7\%},
                        292\comment{2.5\%},
                        300\comment{9.4\%},
                        308\comment{0.9\%},
                        316\comment{1.5\%},
                        324\comment{8\%},
                        332\comment{0.6\%},
                        340\comment{1\%},
                        348\comment{6.7\%}
                        356\comment{0.9\%},
                        364\comment{1\%},
                        372\comment{4.1\%},
                        380\comment{0.8\%},
                        388\comment{4.2\%},
                        396\comment{2.9\%},
                        404\comment{0.3\%},
                        412\comment{1\%},
                        420\comment{2.2\%},
                        428\comment{0.3\%},
                        436\comment{0.8\%},
                        444\comment{1.5\%},
                        452\comment{0.4\%},
                        460\comment{0.6\%},
                        468\comment{0.7\%},
                        1196\comment{0.6\%},
                        1268\comment{0.3\%},
                        1388\comment{0.03\%},
                        1412\comment{0.1\%}   \\
    \hline
    \end{tabular}

  \caption{Highway behaviours from simulations.\\
    This table is read as follows: for ant $LLRL$, in our computations from initial random configurations, we reached a highway behaviour in $22.95\%$ of the computations. Amongst those highway behaviour there is a dominant highway of length $384$ which was reached in $99.9997\%$ of cases and other rare highways of length $308$ (freq. $3\cdot 10^{-4}\%$), $380$ (freq. $4\cdot 10^{-7}\%$) and $928$ (freq. $2\cdot 10^{-7}\%$). }
    \label{tab:stats}
\end{table}

Note that each highway that was found is an actual highway of the corresponding ant, however we make no claim that we found all the possible highways for any of these ants, nor that the proportions we observed are close to the actual probability of reaching the corresponding highway from an initial perturbation of size $n\times n$ taken uniformly at random.

We also tried the same process with another shape of the inital random pattern (for example cross shaped support instead of square support).
In our experience, this did not change significantly the behaviours or frequencies.

From these computations, we draw a few observations.

\paragraph{The $L^+R$ ants behave nicely.}
These ants are the simplest generalised ants and they seem to behave very nicely, though the original $LR$ ant stands out in the behaviour we observe.
Our first observation is that a highway behaviour is reached in every computation we ran.
We then observe two clear subfamilies :
the $L^{2k}R$ ants seem to each have a highway of period $16k+2$ and for $k\geq 2$ a second highway of period $32k+4$;
while the $L^{2k+1}R$ ants seem to each have only a highway of period $32k+20$, but the rule $L^3R$ has two highways, one of them quite infrequent, 3 times longer and with a much bigger support.
One can ask whether other $L^{2k+1}R$ ants have another big and very infrequent highway as $L^3R$ does?

Additionally, we observe that for ant $L^{2k}R$, as $k$ grows, the short (call \emph{fundamental}) highway gets rare and the long (call \emph{harmonic}) highway gets more and more dominant.
This is partly explained by the fact that the harmonic highway is actually $k-1$ distinct highways with the same period as detailed in Section \ref{sec:l2kr}.

\paragraph{The $LLRL$ ant has an overwhelmingly dominant highway, and very rare highways.}
From most initial configurations, no highway behaviour seems to be reached quickly (that is, before $\sim 10^5$ steps). 
Amongst the configuration where a highway is reached there is a dominant highway of length $384$ that is reached in $99.9997\%$ of cases and a few very rare other highways of length $308$, $380$, $388$ and $928$ which all combined are reached in about $0.0003\%$ of cases.
In particular, from the initial $0$-uniform grid, the ant reaches the dominant $384$ period highway after $256100$ steps.
These highways are illustrated on our online simulator\footnote{\href{https://lutfalla.fr/ant/highway.html?antword=LLRL}{\texttt{\small lutfalla.fr/ant/highway.html?antword=LLRL}}.}.

This behaviour that we observe, with a highway which exists but is extremely rare, conforts us in the idea that it just might be possible that the $LR$ ant has other possible asymptotic behaviours but they are so extremely rare that they were never encountered in computations.

\paragraph{The $LLRLRL$ ant has infinitely many highways.}
Among the few initial configurations that reached a highway, it is quite remarkable that no highway is too infrequent, as was the case for rule LLRL.
More remarkably still, we observe arithmetic progressions  in the period of the highways, correlated with some tipicity in the frequencies of these.
In particular we find a family\footnote{\href{https://lutfalla.fr/ant/highway.html?antword=LLRLRL}{\texttt{lutfalla.fr/ant/highway.html?antword=LLRLRL}}.} of highways of length $220+24n$,
though we also find lengths of the form $300+24n$ and $308+24n$.
When looking at the highways them selves, we see that the period does not determine the highway, we find different trajectories  with the same period.
Note that, from the initial $0$-uniform grid, and even after $10^{10}$ steps, the $LLRLRL$ ant does not present any highway behaviour (nor any emerging behaviour). 

This ant, though remarkable, does not seem to be the only ant with infinitely many highways.
Our simulations suggest that the ants with rule words $LLRLRLL$, $LLRLRRLL$, $LLRLRRLL$, $LLRLRRLR$ and $LLRRRLL$ all have infinitely many highways.

\section{Multiple highways for simple family of ants}
\label{sec:l2kr}

The simplest generalised ants are the $L^+R$ ants, that is ants whose rule word has $n$ $L$ states and then a $R$ state.
The classical Langton's ant being the special case for $n=1$.

This family of generalised ants already contains many asymptotic highway behaviours and in this section we are interested in the subfamily $(LL)^+R$, or $L^{2k}R$ ants.

These ants and the highways we present in this section are illustrated on \href{https://lutfalla.fr/ant/highway.html?antword=L2KR&k=4&i=2}{\texttt{\small lutfalla.fr/ant/highway.html?antword=L2KR\&k=4\&i=2}}.
Values for \texttt{k} 
(number of $(LL)$s in the ant word) 
and \texttt{i} (variant of the highway) can be modified in the interface and url. 

\begin{lem}[$L^{2k}R$ cycles] \label{lemma:l2kr_elementary_cycles}
  Let $a,b,c,d \in \mathbb{N}$ such that $2k \geq a \geq b,c,d \geq 0$.
  Let $b' = b+(2k-a)$, $c' = c+(2k-a)$, $d' = d+(2k-a)$.
  Let $P$ and $P'$ be the patterns of Fig.~\ref{fig:l2kr_elementary_cycle}.

  Then $T_{L^{2k}R}^{4(2k-a)}(P) = P'$.

  \begin{figure}[htb]
    \center
    \includegraphics[width=0.3\textwidth]{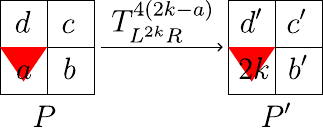}
    \caption{An elementary cycle of the $L^{2k}R$ ant.}
    \label{fig:l2kr_elementary_cycle}
  \end{figure}
\end{lem}

\begin{proof}
   We first remark that the lemma is trivially true when $a=2k$, because in that case $P=P'$ and $4(2k-a)=0$, we can thus assume $a<2k$.
  As we have $a,b,c,d< 2k$, the ant $L^{2k}R$ turns left and stays in the domain $\support{P}$ until it reaches a state $2k$.
  As $a\geq b,c,d$, the first state to become $2k$ and the first $2k$ state to be reached will be in the lower left position (where the state $a$ and the ant initially were).
  Thus, the ant turns left and does $2k-a$ complete left $4$-cycles and reaches configuration $P'$.
\end{proof}

Denote $\cycle{a.b.c.d}=a.b.c.d.a+1.\cdots.2k-1.b+2k-a-1.c+2k-a-1.d+2k-a-1$ the trace of the $L^{2k}R$ ant from $P$ to $P'$, .

\begin{lem}[$L^{2k}R$ almost highways]\label{lemma:l2kr_almost_highway}
  Let $0<i<2k$.

  Let $P_i$ and $P_{2k-i}'$ be the patterns of Fig.~\ref{fig:l2kr_almost_highway_overview}.

  We have $T_{L^{2k}R}^{24k-8i+2}(P_i) = P_{2k-i}'$.

  \begin{figure}[htb]
    \center \includegraphics[width=0.6\textwidth]{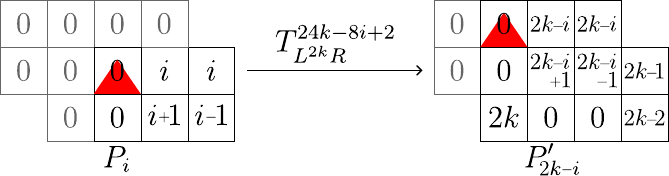}
    \caption{The almost highway of the $L^{2k}R$ ant.}
    \label{fig:l2kr_almost_highway_overview}
  \end{figure}
\end{lem}

\begin{proof}
  We decompose the trace of the $L^{2k}R$ ant starting from $P_i$ in three elementary cycles of lengths $8k$, $8k-4i$, $8k-4(i+1)$ and in six individual moves as depicted in Fig.~\ref{fig:l2kr_almost_highway_decomposition}.

  This means that the trace from $P_i$ to $P_{2k-i}'$ is $t_{k,i}$, where
  \[ t_{k,i} := \cycle{0.0.0.0}.2k.\cycle{i.0.0.0}.2k.\cycle{\cplus{i}{1}.\cmin{i}{1}.i.0}.2k.2k.\cmin{2k}{i}.2k. \]
  
  \begin{figure}[htb]
    \center \includegraphics[width=\textwidth]{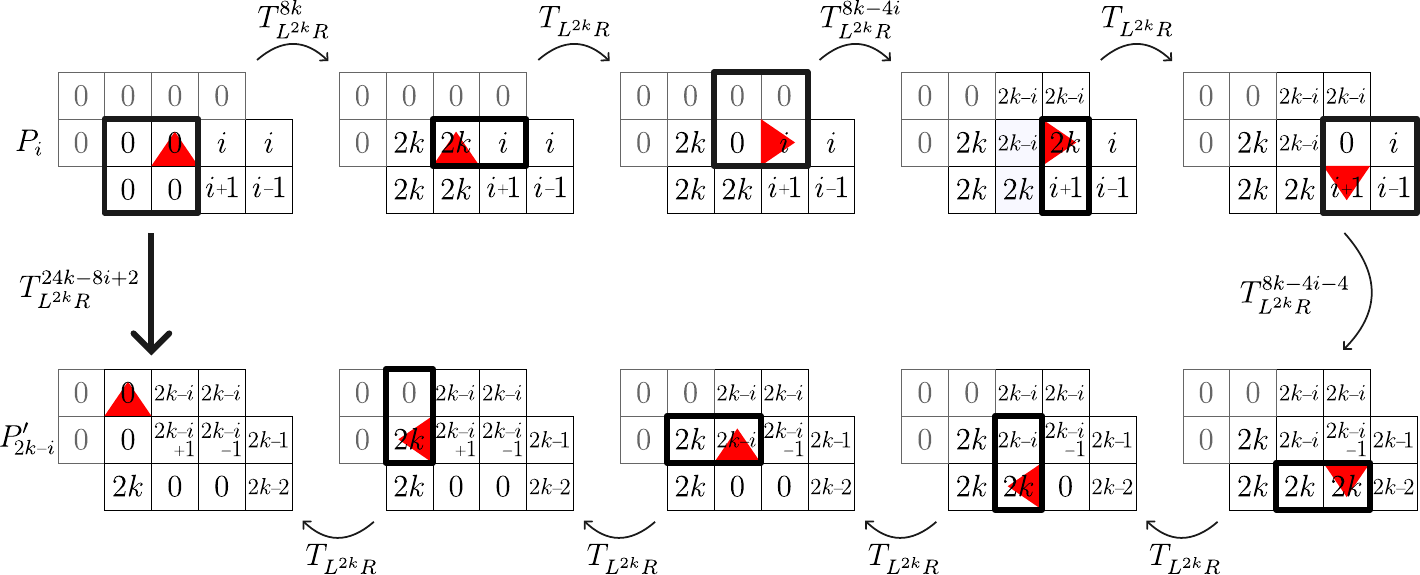}
    \caption{Decomposition of the almost highway for the $L^{2k}R$ ant.\\
      At each step, the domain of the next move or elementary cycle is emphasised with bold boundary.}
    \label{fig:l2kr_almost_highway_decomposition}
  \end{figure}
\end{proof}

Additionally, we remark that for $i\neq j$ the trace 
$t_{k,i}\cdot t_{k,i}$ does not contain the factor subword $j+1.j-1.j.0$ contained in $t_{k,j}$ so none of these almost-highways is a sub-dynamic of another.

Now we explain how the ``almost highways'' can be used to construct complete highways.

\begin{thm}[Fundamental and harmonic highway for $L^{2k}R$]
  \label{th:l2kr_length}
For any integer $k>1$, the $L^{2k}R$ ant admits at least two different highways :
\begin{itemize}
    \item a fundamental highway of period $16k+2$ and drift $(\pm 1, \pm 1)$,
    \item a harmonic highway of period $32k + 4$ and drift $(\pm 2, \pm 2)$.
\end{itemize}
\end{thm}

\begin{proof}
  Let $k>1$ and consider the and $L^{2k}R$.

  For $n\geq 0$, let $C_{k,i}^{n}$ be the $\mathbb{Z}^2$ picture defined as
  \[ C_{k,i}^{n}(x,y)=\left\{ \begin{array}{lll}
      i & \text{ if } & y=n\wedge x\in \{1-n,2-n\} \\
      i+1 & \text{ if } & (x,y)=(1-n,n-1) \\
      i-1 & \text{ if } & (x,y)=(2-n,n-1) \\
      2k & \text{ if } & y = -x -2 \wedge -n \leq x \leq -1\\
      2k-1 & \text{ if } & y = -x+2 \wedge 2 \leq x \leq n+1\\
      2k-2 & \text{ if } & y = -x+1 \wedge 2 \leq x \leq n+1\\
      0 &  & \text{otherwise} 
    \end{array}\right. \]
   
  In particular, $C_{k,k}^{0}$ is the $0$-uniform picture with the pattern $P_k$ of Lemma~\ref{lemma:l2kr_almost_highway} at the origin, $C_{k,k}^{1}$ is the same with pattern $P_k'$, and in general $C_{k,i}^n$ has pattern $P_i$ at position $(-n,n)$.
  We remark that $P_k'$ also has pattern $P_k$ but at position $(-1,1)$.

  We first show the fundamental highway.
  Applying Lemma~\ref{lemma:l2kr_almost_highway} with $i=k$ we obtain that for any $n$,
  \[T_{L^{2k}R}^{16k+2}(C_{k,k}^n, (-n,n), \uparrow) = (C_{k,k}^{n+1}, (-n-1,n+1), \uparrow)\] and the trace of this transition is $t_{k,k}$.
  In other words, $(C_{k,k}^0, (0,0), \uparrow)$ starts a highway of period $16k+2$, drift $(-1,1)$ and trace $t_{k,k}$.
  
  Now we prove the harmonic highway.
  Remark that, by applying Lemma~\ref{lemma:l2kr_almost_highway} with $i=1$ and $i=2k-1$ we have for any $n$ both
  \[ T_{L^{2k}R}^{24k-6}(C_{k,1}^n,(-n,n),\uparrow) = (C_{k,2k-1}^{n+1},(-n-1,n+1), \uparrow)\] with trace $t_{k,1}$ and
  \[T_{L^{2k}R}^{8k+10}(C_{k,2k-1}^{n+1},(-n-1,n+1),\uparrow) = (C_{k,1}^{n+2},(-n-2,n+2), \uparrow),\] with trace $t_{k,2k-1}$.

  Combining the two we obtain
  \[ T_{L^{2k}R}^{32k+4}(C_{k,1}^n,(-n,n),\uparrow) = (C_{k,1}^{n+2}, (-n-2,n+2), \uparrow).\]
  In other words, $(C_{k,1}^0,(0,0),\uparrow)$ starts a highway of period $32k+4$, drift $(-2,2)$ and trace $t_{k,1}\cdot t_{k,2k-1}$.

  The periodic repetition $(t_{k,1}\cdot t_{k,2k-1})^\omega$ has $32k+4$ as a smallest period because the factor subword $2.0.1.0$ only appears once in $t_{k,1}$ and does not appear in $t_{k,2k-1}$; so $32k+4$ is indeed the highway period of this dynamics. 
\end{proof}

But what is done for $i=1$ can be done for every $0<i<k$, thus we have a family of harmonic highways.

\begin{thm}[Variants of the harmonic highway for $L^{2k}R$]
  \label{th:l2kr_traces}
    For any integer $k>1$, there are at least $k-1$ distinct harmonic highways of the $L^{2k}R$ ant.
\end{thm}
\begin{proof}
  Let $k>1$ and define the $\mathbb{Z}^2$ pictures $C_{k,i}^n$ as in the proof of Theorem~\ref{th:l2kr_length}.
  For any $0<i<k$, we have for any $n$, we apply Lemma~\ref{lemma:l2kr_almost_highway} for $i$ and $2k-i$ we obtain:
   \[ T_{L^{2k}R}^{24k-8i+2}(C_{k,i}^n,(-n,n),\uparrow) = (C_{k,2k-i}^{n+1},(-n-1,n+1), \uparrow)\] with trace $t_{k,i}$ and
  \[T_{L^{2k}R}^{8k+8i+2}(C_{k,2k-i}^{n+1},(-n-1,n+1),\uparrow) = (C_{k,i}^{n+2},(-n-2,n+2), \uparrow),\] with trace $t_{k,2k-i}$.
  Combining the two we obtain:
  \[ T_{L^{2k}R}^{32k+4}(C_{k,i}^n,(-n,n),\uparrow) = (C_{k,i}^{n+2}, (-n-2,n+2), \uparrow)\] with trace $t_{k,i}\cdot t_{k,2k-i}$.
  In other words, $(C_{k,i}^0,(0,0),\uparrow)$ starts a highway of period $32k+4$, drift $(-2,2)$ and trace $t_{k,i}\cdot t_{k,2k-i}$.

  These highways are indeed distinct, as mentioned above, since the trace factor subword $i+1.i-1.i.0$ does not appear in $t_{k,j}\cdot t_{k,2k-j}$ for $i\neq j$.
\end{proof}

\begin{rem}
  It is quite remarkable that, for a given $k>1$, all of these $k$ highways (the fundamental and the $k-1$ variants of the harmonic) leave the same print behind consisting of a diagonal of $2k$, a diagonal of $2k-1$ and a diagonal of $2k-2$. See Fig.~\ref{fig:L6R_highways}.
\begin{figure}[htp]
  \center
  \includegraphics[width=0.2\textwidth]{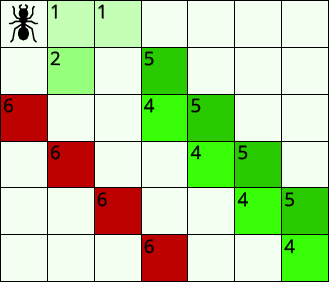} \hfill
  \includegraphics[width=0.2\textwidth]{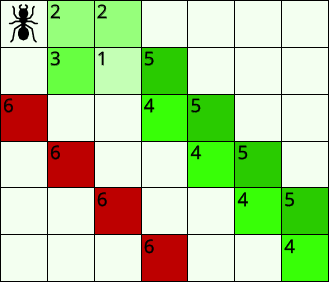} \hfill
  \includegraphics[width=0.2\textwidth]{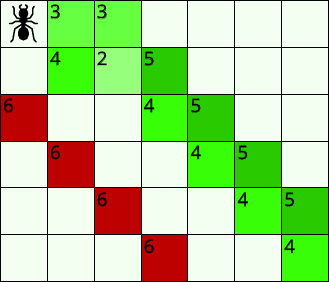}
  \caption{The three highways for the $L^6R$ ant, the rightmost is the fundamental highway.}
  \label{fig:L6R_highways}
\end{figure}
\end{rem}

\begin{rem}[LLR]
  The $LLR$ ant does not have these harmonic highways of length $36$ as there is no $i$ such that $0<i<1$.
  However the fundamental highway of length $18$ described above exists.
\end{rem}

\begin{rem}[$L^{2k+1}R$ ants]
  Analogously, it can be proved that for any $k>1$, the and with rule word $L^{2k+1}R$  has a highway of length $32k+20$.
  Its dynamics, for any \emph{\texttt{k}}, can be seen in
  \href{https://lutfalla.fr/ant/highway.html?antword=L2K1R&k=2}{\emph{\small\texttt{lutfalla.fr/ant/highway.html?antword=L2K1R\&k=2}}}, and its initial pattern is shown in Figure~\ref{fig:l2k+1r}.
  
  \begin{figure}[htp]
    \center \includegraphics[width=0.2\textwidth]{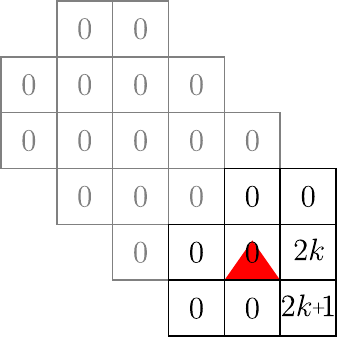}
    \caption{The configuration that starts a highway of period $32k+20$ for the  ant $L^{2k+1}R$.}
    \label{fig:l2k+1r}
  \end{figure}
\end{rem}

\section{Infinitely many highways for a single ant}

\begin{thm}
  The ant $LLRLRL$ has infinitely many highways.\\
  In particular for any $n$, $LLRLRL$ admits a highway of period $220+24n$.
\end{thm}
\begin{proof}
  To prove this result we introduce three widgets, two of them in two versions: a main widget $M_1$ and $M_2$ and a bounce back widget $B_1$ and $B_2$; also a link widget in four versions: $L_1$, $L_2$, $L_3$, and $L_4$.
  
We will prove that, for any $n\in\mathbb{N}$, the pattern $c_n = M_1\cdot L_1^n \cdot B_1$ starts a highway of period $220+24n$, and drift $(-2,-2)$.
Figure~\ref{fig:340} shows a picture of this highway where the widgets can be identified, together with its print leaved after several periods.
The proof is performed by describing the different stages of the ant movement, which can be verified by the reader by hand or by means of any ant simulator, for example \href{https://lutfalla.fr/ant/highway.html?antword=LLRLRL}{\texttt{\small lutfalla.fr/ant/highway.html?antword=LLRLRL}}.

{\bf Start.} The ant start over $c_n=M_1\cdot L_1^n \cdot B_1$ at the red arrow, as Figure~\ref{fig:P1} shows.
The ant will spend 84 steps inside $M_1$, to finally exit and enter the first of the series of $n$ $L_1$s at the position and direction marked by the yellow triangle.
It will take 12 steps to transverse each $L_1$ to finally enter pattern $B_1$ at the position and direction marked by the cyan arrow, where it will spend only 4 steps.
At this point, $M_1$ has been transformed into $M_2$ concatenated with a half $L_2$ and each of the $L_1$s has been transformed into an $L_2$, except for the last one, which slightly differs from the others and will be called $L_2'$.
\begin{figure}[htb!]
\begin{center}
\includegraphics[width=.7\textwidth]{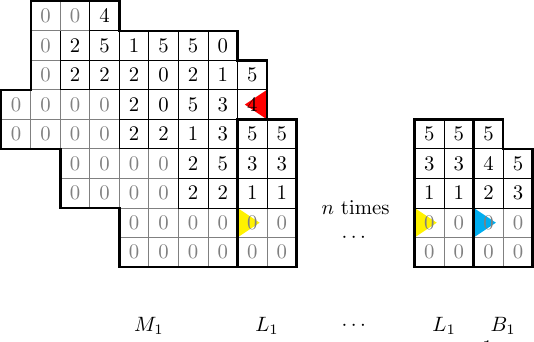}
\caption{Patterns $M_1$, $B_1$, and $L_1$. The red triangle marks the starting ant's position and direction.
The other 3 triangles mark the entering positions of the ant to the respective widget.}
\label{fig:P1}
\end{center}
\end{figure}

{\bf Rebound 1.} Now the ant faces $L_2'$ and a series of $n-1$ $L_2$s, it takes 4 steps to transverse each of them (see Figure~\ref{fig:L2}).
After the ant visit, this series of $n$ widgets is modified.
The column of the first $L_2$, that was adjacent to $B_1$ in the previous stage, will fuse with $B_1$ to become $B_2$.
Starting from the other column of this widget, a new series of $n$ link widgets is formed, that we call $L_3$.

\begin{figure}[htb!]
\begin{center}
\includegraphics[width=0.4\textwidth]{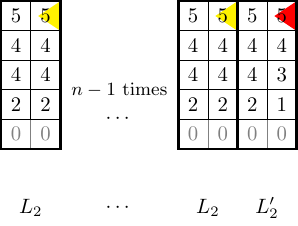}
\caption{Pattern $L_2$. The ant will transverse each block in 4 steps, transforming them into a series of $L_3$s.}
\label{fig:L2}
\end{center}
\end{figure}

{\bf New start.} Now the ant is on $M_2\cdot L_3^n \cdot B_2$ at the position marked by the red arrow (see Figure~\ref{fig:P2}).
The ant will spend 100 steps inside $M_2$, to finally exit it to enter the series of $n$ $L_3$s at the position and direction marked by the yellow triangle.
It will take only 4 steps to transverse each $L_3$ to finally enter pattern $B_2$ at the position and direction marked by the cyan arrow, where it will spend 30 steps.
At this point, $M_2$ has been transformed into $M_1$ concatenated with a half of $L_4$ and each of the $L_3$s has been transformed into an $L_4$, except for the last one, which slightly differs from the others and will be called $L_4'$.

\begin{figure}[htb!]
\begin{center}
\includegraphics[width=.7\textwidth]{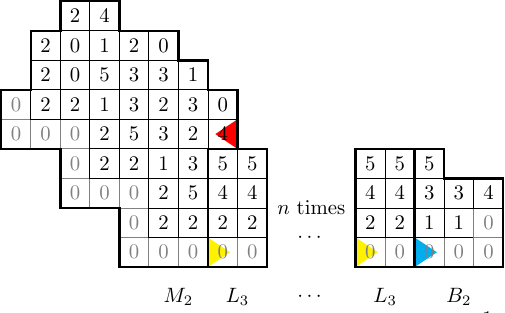}
\caption{Patterns $M_2$, $B_2$, and $L_3$. The red triangle marks the starting ant's position and direction of this stage.
The other 3 triangles mark the entering positions of the ant to the respective widget.}
\label{fig:P2}
\end{center}
\end{figure}

{\bf Rebound 2.} Now the ant faces $L_4'$ and a series of $n-1$ $L_4$s, again it takes 4 steps to transverse each of them (see Figure~\ref{fig:L4}).
After the ant visit, this series of $n$ widgets is modified.
The column of the first $L_4$, that was adjacent to $B_2$ in the previous stage, together with two cells of $B_2$ become $B_1$.
Starting from the other column of this widget, the series of $n$ $L_1$ appears.

\begin{figure}[htb!]
\begin{center}
\includegraphics[width=.4\textwidth]{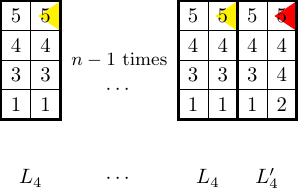}
\caption{Pattern $L_4$. The ant will transverse each block in 4 steps, transforming it into $L_1$.}
\label{fig:L4}
\end{center}
\end{figure}

After these four stages, $c_n$ appears again, but shifted by (-2,-2), covering the gray 0s, and if enough 0s are placed in the background, a new cycle of $220+24n$ steps starts.
\end{proof}

\begin{figure}[htbp]
  \center
  \includegraphics[width=0.5\textwidth]{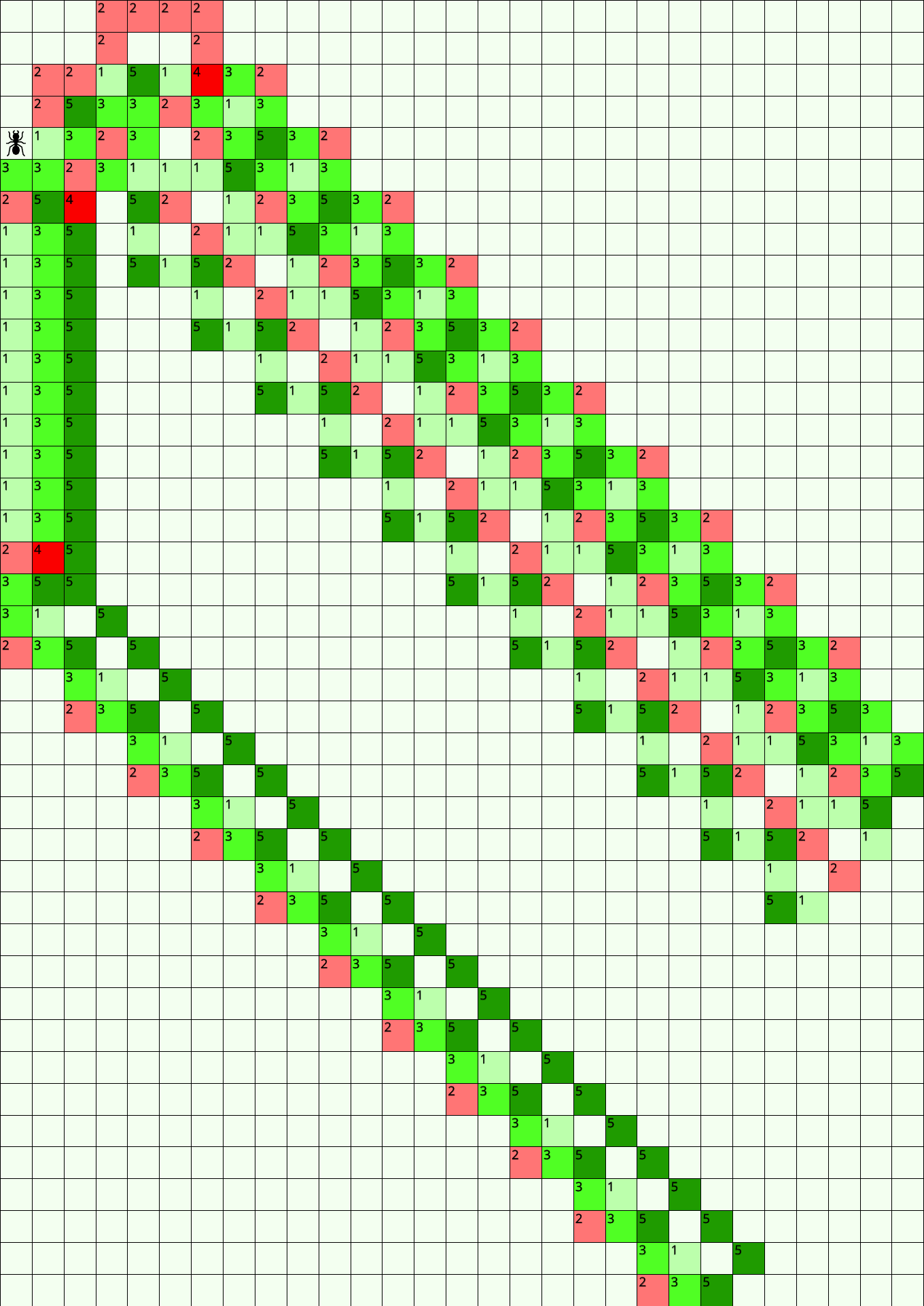}
  \caption{The configuration produced by a highway of period $220+5\cdot 24$ for the $LLRLRL$ ant after $11\cdot(220+5\cdot 24)$ iterations.
    We can identify the $L$ widgets acting as a corridor between the two ``walls'' $M$, on top, and $B$, in the opposite extreme of the series of $L$s.
    The two diagonals on the picture are the prints leaved by widgets $M$ and $B$.\\
  Note that left-turning (resp. right-turning) states are emphasized in green (resp. red). }
  \label{fig:LLRLRL_highway_340}
\end{figure}

\FloatBarrier
 \bibliography{ants}
\end{document}